\documentclass[conference]{IEEEtran}

\PassOptionsToPackage{hyphens}{url}
\usepackage{url}
\usepackage{hyperref}
\usepackage[pdftex]{graphicx}
\usepackage{xcolor}
\usepackage{amsmath}
\usepackage{amssymb}
\usepackage{amsthm}

\newtheorem{theorem}{Theorem}
\newtheorem{corollary}{Corollary}

\begin{document}

\title{Aquaman: A Transparent Proxy Architecture for Quantum Resilient Key Establishment}

\author{
\IEEEauthorblockN{Tushin Mallick\IEEEauthorrefmark{1}}
\IEEEauthorblockA{Cisco Research}

\and
\IEEEauthorblockN{Ashish Kundu}
\IEEEauthorblockA{Cisco Research}

\and
\IEEEauthorblockN{Ramana Kompella}
\IEEEauthorblockA{Cisco Research}
}

\maketitle

\footnotetext[1]{This work was carried out when the author was at Cisco Research, San Jose, CA, during Summer Internship 2025. He is currently at Northeastern University, Boston, MA.}

\begin{abstract}
The \emph{harvest-now, decrypt-later} (HNDL) threat---adversaries
intercepting and archiving ciphertext today for retrospective
decryption once quantum computers mature---turns the future
quantum threat into a present liability for the public-key
primitives (RSA, Diffie--Hellman, ECC) that anchor modern
session-key exchange.

We present \emph{Aquaman}, a transparent-proxy architecture
for quantum-resilient session-key establishment. A transparent
proxy intercepts session-key requests at the edge of a trusted
network without requiring client-side configuration, deploying
quantum-resistant capability at the network boundary on behalf
of clients that may themselves lack post-quantum cryptography
(PQC). Aquaman supports four operating modes: PQC offloaded to
the proxy for clients without trusted PQC stacks; classical
multi-path key fragmentation over heterogeneous media (with an
optional anonymous proxy-pool variant); QKD with the SKIP /
ETSI~GS~QKD~014 key-delivery interface; and classical/PQC
hybrid handshakes. We implement and evaluate the first two
modes; the latter two are well-trodden in the PQC literature
and we discuss but do not implement them. The implemented
multi-path mode splits the session key into ciphertext
fragments distributed across diverse media (Wi-Fi, Bluetooth,
NFC, cellular, Ethernet); reconstruction requires \emph{all}
fragments. We formalize the security argument and prove that
recovery probability decays as $(B/d)^n$ in the diversity
dimension. A 1{,}000-run prototype evaluation on AWS EC2 shows
that latency is dominated by network transmission, not by
multi-path overhead.
\end{abstract}

\section{Introduction}
\label{sec:intro}

Modern secure communication rests on a small set of public-key
primitives---RSA, Diffie--Hellman~(DH), and elliptic-curve
cryptography~(ECC)---whose security is grounded in the conjectured
intractability of integer factorization and discrete logarithms.
Shor's algorithm efficiently solves both problems on a fault-tolerant
quantum computer of sufficient scale~\cite{shor}. While such a
machine does not yet exist, expert consensus increasingly treats its
arrival as an engineering question rather than a theoretical one,
with one-in-three surveyed cybersecurity experts forecasting a
cryptographically relevant quantum computer (CRQC) before
2032~\cite{mosca2018}.

The urgency is sharpened by the \emph{harvest-now, decrypt-later}
(HNDL) threat model: an adversary intercepts and archives ciphertext
today, deferring decryption until quantum capability is available.
For data whose confidentiality must hold for years or
decades---classified communications, medical records, financial
instruments, and critical-infrastructure
telemetry---present-day classical encryption is therefore already
insufficient even though no present-day quantum adversary
exists~\cite{chen2016}. Mosca's inequality $X + Y \ge Z$
formalizes this race between the shelf-life $X$ of the data, the
migration time $Y$ to quantum-resistant systems, and the time $Z$
until a CRQC is realized; many sectors are already operating with
$X + Y > Z$~\cite{mosca2018}.

The National Institute of Standards and Technology (NIST) has driven
a multi-year standardization effort for post-quantum cryptography
(PQC), publishing the first three PQC standards in August
2024---FIPS~203 (ML-KEM), FIPS~204 (ML-DSA), and FIPS~205
(SLH-DSA)~\cite{NIST-pqc-standard-fips203,NIST-pqc-standard-fips204,NIST-pqc-standard-fips205}.
In March 2025, NIST selected HQC as a fifth, code-based KEM, with a
draft standard expected in 2026 and finalization in
2027~\cite{nist-hqc2025}. Draft NIST IR~8547 recommends deprecating
quantum-vulnerable algorithms by 2030 and disallowing them
by 2035~\cite{nist-ir-8547}.

Despite this progress, immediate, universal adoption of PQC is not
operationally realistic in many settings. First, real-world
PQC implementations have already exhibited unexpected weaknesses:
SIKE---a fourth-round NIST candidate---was broken in classical
polynomial time by Castryck and Decru in
2022~\cite{castryck2023}, and lattice-based standards have been
shown to be vulnerable to side-channel and fault-injection
attacks, with multiple chosen-ciphertext key recoveries
demonstrated against Kyber/ML-KEM
implementations~\cite{ravi2023}. Second, large estates of legacy
hardware, embedded controllers, IoT devices, and operational
technology cannot accept new cryptographic primitives without
firmware or hardware replacement, and the lead time for
crypto-agile redesign is measured in years~\cite{nist-ir-8547}.
Third, post-quantum primitives carry larger keys, ciphertexts, and
signatures than their classical counterparts, and these inflations
interact poorly with bandwidth-limited or fragmentation-sensitive
protocols. The cryptographic community's response has been to
deploy \emph{hybrid} schemes that combine a classical KEM with a
PQC KEM and remain secure if either component holds; this is now
the recommended transitional path on the wire~\cite{nist-pqc-transition}.

\textbf{Our Contribution.}
We present \emph{Aquaman}, a transparent-proxy architecture
for quantum-resilient session-key establishment. A transparent
proxy intercepts session-key requests at the network layer, in
contrast to an explicit (non-transparent) proxy that requires
manual client configuration. Deployed at edge devices of a
trusted network, the proxy extends quantum resistance to
clients on the trusted side that may themselves lack
PQC-capable software stacks, and can connect such clients
across an untrusted network. \emph{Aquaman} supports four operating
modes. \textbf{Mode~1 (PQC at the proxy):} clients without
trusted PQC connect classically to the proxy, which runs a PQ
KEM to the QKMS on their behalf. \textbf{Mode~2 (classical
multi-path / anonymous paths):} when no party trusts PQC, the
proxy distributes the key as ciphertext fragments across
heterogeneous media (Wi-Fi, Bluetooth, NFC, cellular,
Ethernet) with all-fragment reconstruction, optionally routed
through an anonymous proxy pool to limit metadata leakage.
\textbf{Mode~3 (QKD key delivery):} the proxy fetches keys
from a co-located QKD device via Cisco's Secure Key Integration
Protocol (SKIP)~\cite{cisco-skip} or ETSI~GS~QKD~014's REST
API~\cite{etsi-qkd-014}; \emph{not implemented in this work}.
\textbf{Mode~4 (classical/PQC hybrid):} the proxy negotiates a
hybrid KEM (e.g., X25519~+~ML-KEM);
\emph{not implemented in this work}. Modes~3 and~4 are
well-trodden in the PQC literature and are included for
deployment completeness; the novel contribution lies in
Modes~1 and~2.

Aquaman can use multi-path mechanism for key establishment for either of these modes. The details of the multi-path protocol including anonymous paths is outside the scope of this paper and is described elsewhere\footnote{Multi-path protocol and anonymous paths are yet to be published; the  references will be updated when they are available}. The security argument for the multi-path mechanism does not
depend on the quantum-hardness of any single algorithm; it
rests on the practical infeasibility of an adversary
simultaneously compromising every channel in a structurally
diverse set. The resulting design (i)~is deployable in
PQ-incapable environments where the wire continues to use
classical encryption, (ii)~stacks safely on top of PQ KEMs as
a defense-in-depth measure, and (iii)~tolerates partial
cryptographic failure on any single channel. We instantiate
the design in a prototype, give an NFC-kiosk bootstrap, and
evaluate it on AWS EC2 across 1{,}000 runs.

The paper is structured as follows. Section~\ref{sec:background} provides background
on the quantum threat, PQC standardization, and secret sharing.
Section~\ref{sec:related} surveys related work in hybrid PQC, QKD
relay networks, and multi-path key establishment.
Section~\ref{sec:threat} formalizes the threat model.
Section~\ref{sec:design} describes the system architecture.
Section~\ref{sec:deployment} details the deployment scenarios and
bootstrap. Section~\ref{sec:eval} reports the empirical
evaluation. Section~\ref{sec:conclusion} concludes.
\section{Background}
\label{sec:background}

This section covers the technical background needed to position our
contribution: the quantum threat to public-key cryptography, the
state of the post-quantum standardization track and its
operational frictions, and secret sharing as a foundation for
threshold-style multi-path schemes. We deliberately omit
introductory material that is well covered in the cryptographic
literature; the focus here is on the points that bear directly on
our design.

\subsection{The Quantum Threat to Asymmetric Cryptography}

Shor's algorithm provides a polynomial-time quantum solution to
both integer factorization and the discrete logarithm problem in
finite fields and on elliptic curves~\cite{shor}. Because RSA, DH,
and ECC reduce, respectively, to instances of these
problems, a sufficiently powerful quantum computer breaks all
three. Symmetric primitives are less affected: Grover's
algorithm yields only a quadratic speed-up against generic key
search, so doubling key sizes (e.g., AES-256) is widely regarded
as adequate~\cite{chen2016}. Consequently, the practical brunt of
the quantum threat falls on the \emph{key establishment} layer of
secure protocols, which is precisely the layer our work targets.

The HNDL threat operationalizes this future capability today.
An adversary collects ciphertext from any classically-keyed channel
and stores it. When a CRQC becomes available, the archived
ciphertext can be decrypted in retrospect. Mosca's inequality
$X+Y \ge Z$ captures the resulting risk: organizations whose data
shelf-life ($X$) plus migration runway ($Y$) exceeds the time
to CRQC ($Z$) are already exposed regardless of when $Z$ actually
arrives~\cite{mosca2018}.

\subsection{Post-Quantum Cryptography and Its Frictions}

NIST's PQC effort has produced four standards as of this writing:
ML-KEM~(FIPS~203), ML-DSA~(FIPS~204), SLH-DSA~(FIPS~205), and a
draft Falcon-based signature standard, with HQC selected as a
fifth, code-based backup KEM in March 2025~\cite{NIST-pqc-standard-fips203,NIST-pqc-standard-fips204,NIST-pqc-standard-fips205,nist-hqc2025}.
ML-KEM (the standardized form of CRYSTALS-Kyber) is now widely
deployed in hybrid TLS handshakes alongside X25519 or NIST
P-curves~\cite{kemtls}.

Three frictions matter for our design:

\textbf{(i) Cryptanalytic risk on new primitives.} Lattice
hardness has been studied for decades, but the deployed PQC
candidates still represent a comparatively young attack surface.
The polynomial-time break of SIKE---a fourth-round NIST
candidate---by Castryck and Decru in
2022~\cite{castryck2023} is a salient reminder that confidence
in new mathematical assumptions can collapse abruptly. Even
where the underlying problem appears robust, implementations
have proven fragile: chosen-ciphertext, EM, and power
side-channel attacks have recovered ML-KEM and Dilithium keys
from masked and unmasked
implementations~\cite{ravi2023}. NIST's PQC migration guidance
explicitly endorses hybrid schemes during the transition for
exactly this reason~\cite{nist-pqc-transition}.

\textbf{(ii) Hardware and protocol inertia.} PQC primitives
typically have larger public keys, ciphertexts, and signatures
than RSA or ECC. ML-KEM-768 ciphertexts are over a kilobyte;
SLH-DSA signatures are tens of kilobytes. These sizes interact
poorly with constrained protocols (DNSSEC, embedded TLS) and with
firmware-frozen IoT and operational-technology
hardware~\cite{nist-ir-8547}. Crypto-agility---the ability to
swap primitives without redesign---has accordingly become a
first-class design goal, but it is itself a multi-year program for
large enterprises.

\textbf{(iii) Standards trail deployment.} NIST IR~8547 sets a
deprecation horizon of 2030 for $\le$112-bit-strength quantum-vulnerable
algorithms (e.g., RSA-2048, P-256) and full disallowance by
2035~\cite{nist-ir-8547}. The window between today and that
horizon---during which classically-keyed traffic will still be
generated---is precisely the window the HNDL adversary exploits.

\subsection{Secret Sharing and Multi-Path Distribution}

The classical primitive for splitting a secret into independently
useless shares is Shamir's $(t,n)$-threshold secret sharing
scheme~\cite{shamir1979}: any $t$ shares reconstruct the secret,
while fewer than $t$ shares reveal nothing in an
information-theoretic sense. Our design uses a $(n,n)$-style split
in spirit---all fragments are needed for reconstruction---rather
than a true polynomial-interpolation Shamir scheme, because (a)~we
do not require partial-availability tolerance, and (b)~the
heterogeneous-channel adversarial model already provides the
unavailability guarantee we need.

Multi-path transport is a second relevant ingredient. Network
research has long exploited path diversity for throughput and
resilience (multipath TCP, equal-cost multi-path, and so on), and
several lines of work have applied secret-sharing-style splitting
\emph{across} such paths to obtain confidentiality even when an
adversary controls some
paths~\cite{barnett2012}. Our contribution
in this lineage is the deliberate use of \emph{transmission-medium
heterogeneity}---Wi-Fi vs. Bluetooth vs. NFC vs. cellular, etc.,
not just disjoint IP routes---as the source of adversarial cost.
Compromising a Wi-Fi link and a Bluetooth link and an NFC reader
simultaneously requires fundamentally different attack
capabilities and proximities, and an adversary's marginal cost of
adding a new medium is much higher than the cost of adding a new
IP path.

\subsection{Why Not QKD?}

Quantum key distribution (QKD), originating with the BB84
protocol~\cite{bb84}, achieves information-theoretic key
establishment by exploiting the no-cloning theorem and
measurement disturbance. QKD is genuinely
quantum-secure---its security does not rest on any computational
assumption. However, current QKD systems are limited by
distance (typically a few hundred kilometers without trusted
relays), specialized hardware (single-photon sources and
detectors), and cost. Most enterprise deployments have neither
the optical infrastructure nor the budget for QKD. Our design
explicitly targets the much larger population of deployments
that must achieve quantum resilience over commodity
\emph{classical} transmission media.

\section{Related Work}
\label{sec:related}

Seven lines of prior work are relevant to a transparent-proxy
architecture for quantum-resilient session-key establishment:
hybrid classical/PQC constructions, secret sharing on QKD
relay networks, multi-path and multi-interface network
security, threshold and multi-secret sharing, anonymous routing
and mix networks (relevant to the proxy-pool variant),
post-quantum migration of anonymity networks, and post-quantum
cryptography deployed at the network edge. We describe each in
turn.

\subsection{Hybrid Classical/PQC Constructions}

The mainstream operational response to cryptanalytic uncertainty
around new PQC primitives is the \emph{hybrid} construction:
combine a classical KEM (e.g., X25519) with a PQ KEM (e.g.,
ML-KEM) so that the derived shared secret is secure as long as
either component is unbroken. Schwabe, Stebila, and
Wiggers~\cite{kemtls} introduced KEMTLS, which replaces
signature-based authentication in the TLS~1.3 handshake with
KEMs and is structurally amenable to such hybridization.
Stebila and Mosca~\cite{oqs2017} introduced the Open Quantum
Safe project, which packages liboqs (a C library of
quantum-resistant primitives) together with prototype
integrations into OpenSSL and other widely deployed
cryptographic stacks; OQS has subsequently served as the
substrate for many of the experimental hybrid TLS deployments
in the literature. Production deployments have followed:
Cloudflare, Google, AWS, and Meta have rolled out ML-KEM/X25519
hybrid handshakes in TLS~1.3, and NIST's migration
guidance~\cite{nist-pqc-transition} explicitly endorses hybrid
constructions during the transitional period. Hybrids provide
an \emph{algorithmic} hedge against the failure of a single
cryptographic primitive while leaving the surrounding protocol,
endpoint, and channel structure unchanged.

\subsection{Secret Sharing on QKD Relay Networks}

Barnett and Phoenix~\cite{barnett2012} described how to secure
a quantum-key-distribution relay network by treating
intermediate relays as trusted eavesdroppers and applying
secret sharing across multiple physical paths through the
relay graph. An adversary recovering the shared key must
compromise every relay on every path; the construction
leverages the information-theoretic security of QKD on each
link together with the combinatorial requirement of all-path
compromise. The line of work assumes specialized photonic
infrastructure and QKD-capable relays, and is targeted at
high-assurance settings such as metropolitan-scale QKD
backbones.

\subsection{Multi-Path and Multi-Interface Network Security}

Multipath TCP (MPTCP) and analogous transport-layer protocols
distribute application data across multiple network paths for
throughput and resilience; as a side effect, an on-path
observer of a single subflow recovers only a fragment of the
data. Subsequent proposals harden MPTCP by applying separate
VPN tunnels (e.g., per-subflow WireGuard) so that interception
of any one subflow leaks only encrypted, partial information.
In the information-theoretic literature, post-quantum-secure
and reliable transmission over heterogeneous networks has been
achieved through linear network coding across multiple noisy
links, given public-key encryption available on a subset of
trusted links. The thrust of this literature is bulk-data
transmission over an essentially homogeneous IP fabric
(multiple paths through the same Internet), with security
arising from path multiplicity rather than from medium
diversity.

\subsection{Threshold and Multi-Secret Sharing}

Shamir~\cite{shamir1979} introduced $(t,n)$-threshold secret
sharing, in which a secret is split into $n$ shares such that
any $t$ shares recover it but fewer than $t$ reveal nothing
about it. Subsequent work has produced verifiable, proactive,
information-theoretically private, and post-quantum-secure
variants targeted at distributed storage, threshold
cryptography, and blockchain-style settings. These
constructions optimize for share verifiability, partial
availability under share loss, and per-share storage cost in a
model where shares reside at independent storage nodes for
extended periods.

\subsection{Anonymous Routing and Mix Networks}

Chaum~\cite{chaum1981} introduced the mix network: a sender
nests encryptions, one layer per relay, so that each relay
strips a single layer and forwards the result, and no single
relay sees both source and destination. The construction
provides unlinkability between sender and recipient under the
assumption that at least one mix in the cascade is honest.
Dingledine, Mathewson, and Syverson~\cite{tor2004} adapted
this idea to interactive, low-latency traffic in Tor, using
telescoping circuits, perfect forward secrecy, directory
servers, and configurable exit policies to provide anonymity
against an adversary observing or controlling a fraction of
relays. Standard analyses bound the probability that an
adversary controlling $a$ of $P$ relays correlates an entry/exit
pair as approximately $(a/P)^{2}$, under the assumption that
hops are sampled approximately independently---the same bound
that arises in the proxy-pool analysis of
Section~\ref{sec:pool-formal}.

\subsection{Post-Quantum Migration of Anonymity Networks}

Berger, Lemoudden, and Buchanan~\cite{torpq2025} survey the
cryptographic primitives used at each layer of the Tor stack
and quantify the performance overhead of replacing the ntor
circuit-extension handshake with a hybrid KEM-based
construction (e.g., ML-KEM combined with X25519). Their
benchmarks isolate classical-cryptography time on constrained
devices and project PQC overhead by substituting post-quantum
primitive timings into the same circuit. The work focuses on
data-path cryptography of an existing anonymity network rather
than on the key-establishment phase itself.

\subsection{Post-Quantum Cryptography at the Network Edge}

A growing body of work places post-quantum cryptographic
capability at the network edge---in CDN POPs, cloud gateways,
or trusted-network boundary devices---rather than at every
endpoint. Industrial deployments of edge PQC TLS by Cloudflare,
Google, AWS, and Akamai are documented in their respective
engineering blogs and use the same hybrid X25519~+~ML-KEM
ciphersuites discussed in the previous subsections;
academically, Amiriara, Mirmohseni, and
Tafazolli~\cite{amiriara2025} formalize an edge-computing PQC
framework for resource-constrained IoT devices, in which
devices offload heavyweight cryptographic operations to a
post-quantum edge server (PQES) and use physical-layer
techniques (wiretap coding, friendly jamming) to protect the
device-to-edge link during offload. The motivation in this
line of work is that the computation and bandwidth cost of
PQC primitives is prohibitive on constrained devices, but
manageable when concentrated at a small number of edge nodes;
the edge node thereby acts as a cryptographic gateway for the
clients behind it.
\section{Threat Model and Assumptions}
\label{sec:threat}

We consider a networked setting in which two trusted endpoints
wish to establish a session key with the assistance of a
trusted Quantum Key Management Server (QKMS), in the presence
of an adversary that combines classical network-level
capabilities with future quantum cryptanalytic capabilities.

\subsection{Adversary Model}

\textbf{Network observation.} The adversary can passively
observe traffic at one or more points in the network, capturing
ciphertext and metadata. This subsumes the
\emph{harvest-now, decrypt-later} (HNDL) capability: archived
ciphertext can be retained indefinitely and decrypted retroactively
when computational capability permits.

\textbf{Quantum cryptanalysis.} The adversary is assumed to have,
either now or at some future point of decryption, sufficient
quantum computational capability to break any classical asymmetric
primitive whose security reduces to integer factorization or
discrete logarithm (RSA, DH, ECC). For symmetric primitives, we
assume Grover-bounded attacks only; AES-256 and SHA-2/3 at
$\ge 256$ bits remain quantum-secure under the standard
assumptions~\cite{chen2016}.

\textbf{Limited-but-non-trivial PQC compromise.} In the PQ-capable
deployment scenario (Section~\ref{sec:deployment}), we further
assume the adversary may know an exploitable weakness in the
\emph{specific} post-quantum implementation deployed on the wire
(e.g., a side-channel or chosen-ciphertext attack against a Kyber
implementation~\cite{ravi2023}, or an isogeny-style polynomial-time
break analogous to Castryck--Decru~\cite{castryck2023}). The
adversary does not, however, possess a generic break of the PQC
algorithm family; if it did, the protection budget shifts entirely
to the multi-path layer.

\textbf{Channel compromise (capacity vector).} The adversary's
per-medium surveillance capability is represented by a vector
$\mathbf{c} = (c_1, \ldots, c_d) \in [0,1]^d$ over the
$d$ medium types in use, where $c_i$ is the probability that
the adversary intercepts and decrypts an arbitrary fragment
transmitted on a channel of type $i$. Compromising a Wi-Fi
network, a Bluetooth pairing, an NFC reader, a cellular
basestation, or a wired Ethernet segment each requires distinct
attack capabilities, distinct physical proximity or
infrastructure access, and distinct operational footprints;
critically, capability on one medium type does not transfer to
another (an SDR within range of a Wi-Fi link does not help
compromise an NFC reader). We accordingly model per-type
compromise events as independent. The adversary's overall
resources are bounded by a budget constraint
$\sum_{i=1}^{d} c_i \le B$ for some $B \in (0, d)$; this is the
central operational assumption of the scheme.
Section~\ref{sec:heterogeneity} formalizes the model and shows
that recovery probability decays as $(B/d)^n$ in the diversity
dimension $d$.

\subsection{Out of Scope}

We do not defend against:
(a)~compromise of the QKMS itself or of either endpoint's
private key;
(b)~an adversary capable of \emph{simultaneously} compromising
\emph{every} channel in the chosen heterogeneous set;
(c)~denial-of-service attacks on the channel set sufficient to
prevent reconstruction (these reduce availability but not
confidentiality);
(d)~side-channel attacks on the endpoint hardware
(timing, power, EM emanations).

\subsection{Trust Assumptions}

The QKMS is trusted with respect to confidentiality of the
generated session key (it holds the key in plaintext at
generation time). Its operator is therefore at minimum
on the same trust footing as a certificate authority or a
key-management appliance in a current enterprise PKI. The optional
transparent proxy is trusted only to forward fragments
faithfully; it never sees the key in plaintext, since fragments
are individually encrypted under the recipient's classical public
key before they leave the QKMS.

\section{System Design}
\label{sec:design}

The system has three logical components: (i)~a Quantum Key
Management Server (QKMS) that generates session keys and
distributes them as encrypted fragments across heterogeneous
channels; (ii)~the communicating endpoints (clients) that
request keys and reconstruct them from fragments; and (iii)~a
transparent proxy that mediates between client and QKMS,
intercepting client traffic at the network layer without
requiring client-side configuration. The proxy is the central
architectural primitive of \emph{Aquaman}: it allows
quantum-resistant capability to be deployed at the boundary of
a trusted network and extended to clients on the trusted
side---including clients that themselves lack PQC-capable
software stacks. Each component exposes its functionality as
a small set of HTTP endpoints implemented in Flask.

\subsection{Quantum Key Management Server (QKMS)}

The QKMS accepts a structured key-request message and replies
with a session key whose generation parameters are dictated by
the requester. The current implementation defines five
parameters:

\textbf{tagname.} A unique identifier mutually agreed
upon out-of-band by the two communicating clients. Key generation
proceeds only when the tagnames presented by both endpoints match;
this binds a key issuance to a specific session and prevents
unsolicited fragment delivery.

\textbf{key type.} The bit-length of the symmetric session key
to be generated. The key itself is sampled uniformly at random
from the chosen length.

\textbf{number of splits.} The number of fragments $n$ into which
the key is partitioned. Reconstruction requires \emph{all} $n$
fragments; the scheme does not use a $(t,n)$ threshold.

\textbf{shuffle.} A boolean flag instructing QKMS to randomly
permute fragment order before transmission. Combined with the
per-fragment encryption (below), this prevents an adversary that
captures a subset of fragments from inferring their position in
the original key.

\textbf{channels.} A list of available transmission media,
selected from a heterogeneous set: Wi-Fi, Bluetooth, NFC,
cellular, Ethernet, and distinct logical ports on the same
device. The QKMS draws fragments and channels independently and
assigns each fragment to a channel chosen uniformly at random;
some channels may carry zero fragments, others more than one.

After fragmentation and (optional) shuffling, each fragment is
individually encrypted under the recipient's classical public key.
The resulting ciphertext fragments are then dispatched across the
selected channels. The security argument (formalized in
Section~\ref{sec:deployment}) is that an adversary must
simultaneously compromise every channel that carries at least one
fragment in order to recover even \emph{one} fragment in
plaintext, since the per-fragment encryption hides the fragment
even from a passive same-channel observer; and the adversary must
recover \emph{every} fragment in order to reconstruct the
key. The heterogeneity of the channel set is the practical lever
that makes simultaneous total compromise infeasible.

QKMS exposes a single \texttt{POST /get-key-parameters} endpoint.
The request body carries the parameters above along with the the requester's
public key; this triggers fragment generation and dispatch.

\subsection{Transparent Proxy}

The proxy decouples the client from the QKMS and is the
mechanism by which \emph{Aquaman} supports clients that cannot, or
should not, contact the QKMS directly. Two configurations are
supported. As a \emph{transparent proxy}, it is deployed at an
edge router or gateway of a trusted network and intercepts
client traffic at the network layer without requiring the
client to configure proxy settings; this is the deployment we
emphasize, since it allows quantum-resistant capability to be
added to a network without modifying the client devices it
serves. As an \emph{explicit proxy}, the client receives the
proxy's address (e.g., during the bootstrap step of
Section~\ref{sec:deployment}) and dispatches its traffic to it
directly. In both configurations the proxy supports the same
operating modes (Section~\ref{sec:deployment}). Two motivations
apply: (i)~the client may wish to remain anonymous to the QKMS;
(ii)~the client may lack the multi-interface capability to
expose heterogeneous channels (e.g., a cloud workload reachable
only via a single Internet path). In the latter case, the
proxy can expose its own heterogeneous channel set on the
client's behalf.

The proxy receives the client's payload (key parameters and public
key) and forwards it to the QKMS together with the proxy's own
channel list. Encrypted fragments returned over the
heterogeneous channels are aggregated at the proxy and forwarded
to the client over a single, ordinary IP path; because the
fragments are encrypted under the \emph{client's} public key (not
the proxy's), the proxy never sees the key in plaintext. The
proxy exposes a \texttt{POST /get-key-parameters} endpoint
mirroring the QKMS's, and forwards each received fragment to the
client's \texttt{POST /receive-key-fragment} endpoint.

\subsection{Communicating Party (Client)}

The client initiates the session-key acquisition. It either
contacts the QKMS directly (when it has multi-interface
capability) or contacts a proxy. Its outgoing payload contains:
\textit{key type}, \textit{number of splits}, \textit{tagname},
\textit{shuffle}, and the client's classical \textit{public key}.
The client's Flask application exposes either
\texttt{POST /receive-key-fragment} (to receive fragments from a
proxy) or the per-channel listening endpoints (to receive
fragments directly from the QKMS). Upon receipt, the client
decrypts each fragment with its private key, validates that the
expected number of fragments has arrived, undoes any shuffle
order, and concatenates the fragments to reconstruct the session
key. The reconstructed key is then used for symmetric encryption
of the application traffic (e.g., AES-256-GCM).
\section{Deployment Scenarios}
\label{sec:deployment}

The high-level problem we address: two parties wish to
communicate securely and trust a QKMS to provide session keys in
a quantum-resilient manner. We organize the discussion by the
PQC capability of the entities involved, with both
configurations preceded by a common bootstrap.

\subsection{Bootstrap}

A client begins by physically approaching a designated kiosk and
presenting a pre-registered badge that authenticates over Near
Field Communication (NFC). On successful authentication, the
kiosk issues a QR code containing three fields:
\begin{enumerate}
\item the IP address of a temporary proxy node that will mediate
between the client and the QKMS;
\item an expiration timestamp bounding the validity window of the
key request;
\item a digital signature, computed under the kiosk's private key,
over the previous two fields.
\end{enumerate}
The signature ensures integrity and authenticity---in particular
preventing a client from extending its authorized window---and
the expiration timestamp limits the temporal scope of the
issued credential. If the client elects to use a proxy, it
forwards the QR-code contents to the chosen proxy. We assume the
client--proxy link is itself secure (e.g., direct LAN, or
classically-tunneled over a single trusted network); this is a
realistic assumption when the proxy is operated by the same
organization as the kiosk.

The bootstrap is needed because the QKMS's location must reach the
client in a quantum-safe manner before any cryptographic exchange
takes place. NFC's short range (typically a few centimeters) is
the relevant security property: an HNDL adversary cannot harvest
NFC traffic at scale or remotely, and the air-gapped nature of
the kiosk interaction places this step outside the over-the-Internet
threat surface.

\subsection{Scenario 1: PQ-Incapable Environment}

Neither the clients nor any proxies have post-quantum
capabilities; all wire-line cryptography is classical.
Consequently every link is vulnerable to a future quantum
adversary or to HNDL collection.

\paragraph*{Multi-path between QKMS and client/proxy.}
If the client uses a proxy, the proxy is treated as trusted from
the client's perspective. The proxy exposes its heterogeneous
channel set to the QKMS on the client's behalf. Communication
between the client and the proxy is a direct IP-based channel
(no anonymization). The QKMS encrypts each fragment under the
client's classical public key and dispatches the fragments across
the proxy's heterogeneous channels; the proxy aggregates and
forwards them to the client.

\paragraph*{Private proxy pool between client and QKMS.}
A more conservative variant assumes the client controls a private
pool of proxy nodes. For each request, the client randomly
selects an \emph{entry node} from the pool and forwards its
payload there. Each proxy that receives the payload appends its
own channel list to the payload and then independently decides
either (a)~to forward the payload to another proxy in the pool or
(b)~to deliver it directly to the QKMS. The proxy that ultimately
contacts the QKMS is the \emph{exit node}.

A parameter $\mathrm{MAX\_HOPS}$ bounds the number of intra-pool
proxy hops. The forwarding decision at each hop may be uniformly
random or weighted (e.g., a probability of forwarding that
decreases with the hop count, biasing the path toward a finite
length). Once the QKMS has received the payload, it dispatches
encrypted fragments either to the entry node or to the exit
node via the multi-path mechanism. If sent to the entry node,
fragments are forwarded to the client. If sent to the exit
node, the fragments traverse a similar pool-routed path back to
the entry node before being forwarded to the client.

This pool variant adds an anonymity layer: no single proxy in
the pool sees the full request flow, and the entry/exit role
is randomized per request.

\textbf{Formalization of Proxy-Pool Anonymity. }
\label{sec:pool-formal}
The proxy-pool variant adds a metadata-privacy layer that is
logically independent of the multi-path secrecy layer of
Section~\ref{sec:heterogeneity}. We make this independence
explicit: the pool contributes \emph{nothing} to the
cryptographic secrecy of the established session key. Its role
is to constrain what an external observer or the QKMS itself
can learn about \emph{which client is interacting with the
QKMS}. We state the (modest) guarantees the pool does provide
and---equally importantly---the guarantees it does not.

\paragraph*{Setup.} The client $C$ controls a pool
$\mathcal{P}$ of $P=|\mathcal{P}|$ proxies. For each request,
$C$ selects an entry $e\in\mathcal{P}$ uniformly at random.
Each proxy receiving the payload independently forwards to a
uniformly chosen successor in $\mathcal{P}$ with probability
$q$ or delivers to the QKMS with probability $1-q$, subject to
the hop cap $h_{\max}=\mathrm{MAX\_HOPS}$. The number of hops
$h$ on a request is a truncated geometric random variable with
$\mathbb{E}[h]\le \min\!\big(1/(1-q),\,h_{\max}\big)$.

\paragraph*{Adversary.} Two adversary classes are relevant,
both consistent with the threat model of
Section~\ref{sec:threat}:
(i)~the QKMS itself (or anyone restricted to the QKMS's view),
which observes only the exit;
(ii)~a network observer $\mathcal{A}$ that surveils a fixed
subset $S\subseteq\mathcal{P}$ of size $a$, meaning it can read
inter-proxy traffic on those proxies---including under quantum
capability that defeats the classical hop-to-hop encryption.

\paragraph*{Claim 1 (anonymity from the QKMS).} The QKMS sees
only the exit identity $x$, not the entry $e$ or $C$. Under
uniform forwarding, $x$ is a uniform draw from $\mathcal{P}$
that is independent of $C$'s identity. The QKMS therefore gains
no information about which client originated the request beyond
``some client controlling some proxy in $\mathcal{P}$''. This
guarantee is unconditional---it does not rest on the security
of any cipher.

\paragraph*{Claim 2 (path-tracing bound).} For an
$\mathcal{A}$ surveilling $S$, the probability that every hop
along a request path of length $h$ lies in $S$ is, for hops
drawn approximately independently from $\mathcal{P}$,
\[
\Pr[\,\text{path}\subseteq S\,] \;=\; (a/P)^{h}.
\]
This bounds the probability that $\mathcal{A}$ traces the
request hop-by-hop back to the entry by stepping through
observed proxies. The bound decays exponentially in $h$.

\paragraph*{Claim 3 (entry-exit correlation bound).}
Linking a client-to-entry observation with an exit-to-QKMS
observation by traffic correlation requires $\mathcal{A}$ to
observe both the entry $e$ and the exit $x$. Treating $e$ and
$x$ as approximately independent uniform draws (exact for
$h\ge 2$ in the limit of large $P$),
\[
\Pr[\,e\in S \,\wedge\, x\in S\,] \;\approx\; (a/P)^{2}.
\]
This matches the standard onion-routing
analysis~\cite{tor2004}.

\paragraph*{What the pool does \emph{not} buy.}
The pool offers no additional protection against the recovery
of the session key itself. An adversary capable of compromising
the multi-path layer of Section~\ref{sec:heterogeneity} is not
weakened by the pool's presence, and an adversary incapable of
compromising that layer already fails regardless of the pool.
The two layers are orthogonal: multi-path protects \emph{what is
said}, the pool protects \emph{who is saying it} to the QKMS.
Stacking the pool atop multi-path therefore yields a
two-property defense (secrecy~$\times$~metadata privacy), not a
single, stronger secrecy property.

\paragraph*{Limitations.} The bounds presume:
(i)~hops are sampled approximately independently from
$\mathcal{P}$, as required for the product-form probabilities;
(ii)~$\mathcal{A}$'s surveillance set $S$ is fixed in advance
of the request and not adapted on observed routing;
(iii)~no timing- or volume-based side channel allows
$\mathcal{A}$ to correlate ingress and egress of an unobserved
hop without being on its wire. Side-channel correlations
generally require cover traffic or batched
mixing~\cite{chaum1981} to defeat, and we add neither; the
classical mix-network and Tor-style traffic-analysis literature
applies here unchanged and constrains what the pool can
guarantee.

\subsection{Scenario 2: PQ-Capable Environment}

The endpoints lack quantum capability themselves but employ
PQ-capable proxies for key exchange. Crucially, the deployed
PQC implementation may itself be vulnerable---e.g., to a
side-channel attack against ML-KEM or to a future algorithmic
break analogous to that on SIDH. The bootstrap is omitted in this
scenario because the explicit goal is to provide
quantum-resilient security \emph{without} relying on PQC alone.

\paragraph*{Multi-path on top of PQ KEM}
A post-quantum KEM (e.g., ML-KEM) is established between the
client/proxy and the QKMS. The shared secret derived from the PQ
KEM is used as a symmetric key (AES) protecting the subsequent
exchange. The client forwards key parameters and its classical
public key to the QKMS through this PQ-protected channel,
optionally via a proxy. The QKMS then dispatches the key
fragments---each encrypted under the client's classical public
key---across the heterogeneous channels.

The two layers are structurally independent: the wire-line
encryption uses ML-KEM, while the per-fragment payload encryption
uses RSA/ECC. To recover a fragment, the adversary must break
\emph{both} the ML-KEM channel and the per-fragment classical
encryption \emph{on the same channel}; to recover the key, this
must succeed for every fragment-bearing channel. Even if the
deployed ML-KEM implementation is later found to be vulnerable,
the multi-path layer continues to require simultaneous
compromise of every heterogeneous medium that carried a fragment.

\subsection{A Formal Capacity Model}
\label{sec:heterogeneity}

We formalize the heterogeneity argument as follows. Let
$\mathcal{T} = \{1, \dots, d\}$ index the distinct \emph{medium
types} (Wi-Fi, Bluetooth, NFC, cellular, Ethernet, etc.); we
call $d$ the \emph{diversity dimension}. The QKMS partitions
the session key into $n$ fragments and assigns $n_i$ fragments
to type $i$, with $\sum_i n_i = n$. The adversary's per-type
capability is the capacity vector $\mathbf{c} \in [0,1]^d$
introduced in Section~\ref{sec:threat}, subject to a budget
constraint
\begin{equation}
\sum_{i=1}^{d} c_i \;\le\; B,
\qquad B \in (0, d).
\label{eq:budget}
\end{equation}
Assuming per-type compromise events are independent and that
the per-fragment classical encryption is broken on every
surveilled channel (worst case for the multi-path layer), the
recovery probability is
\begin{equation}
P_{\mathrm{rec}}(\mathbf{c}, \{n_i\})
\;=\; \prod_{i=1}^{d} c_i^{\,n_i}.
\label{eq:rec}
\end{equation}

\begin{theorem}[Recovery decays in the diversity dimension]
\label{thm:decay}
With uniform fragment distribution $n_i = n/d$, the maximum
recovery probability over capacity vectors satisfying
\textnormal{(\ref{eq:budget})} is
\[
P_{\mathrm{rec}}^\star(d) \;=\; \left( \frac{B}{d} \right)^{n}.
\]
\end{theorem}

\begin{proof}
With $n_i = n/d$, equation~(\ref{eq:rec}) becomes
$P_{\mathrm{rec}} = (\prod_i c_i)^{n/d}$. By AM--GM,
$\prod_i c_i \le (\sum_i c_i / d)^d \le (B/d)^d$, with
equality at $c_i = B/d$. Hence
$P_{\mathrm{rec}}^\star = ((B/d)^d)^{n/d} = (B/d)^n$.
\end{proof}

\noindent
\emph{Defender-optimality of uniform $n_i$.} For an arbitrary
allocation, a Lagrangian computation gives
$P_{\mathrm{rec}}^\star(\{n_i\}) = (B/n)^n \prod_i n_i^{n_i}$,
which is minimized over $\sum n_i = n$ at $n_i = n/d$ by
strict convexity of $x \log x$. Uniform fragment distribution
is therefore the defender's minimax strategy.

\begin{corollary}[Required diversity]
\label{cor:diversity}
$P_{\mathrm{rec}}^\star \le \varepsilon$ whenever
$d \ge B \cdot \varepsilon^{-1/n}$.
\end{corollary}

\noindent
\emph{Example.} With $n = 8$ fragments and adversary budget
$B = 2$ (the adversary can fully saturate two medium types
in expectation), $d = 4$ yields
$P_{\mathrm{rec}}^\star \le 2^{-8}$---a configuration
achievable on a commodity smartphone exposing Wi-Fi,
Bluetooth, NFC, and cellular interfaces. Doubling to $d = 8$ tightens the bound
to $2^{-16}$.

\noindent
\emph{Convex-cost generalization.} For per-type costs
$\phi(c) = c^k$ with $k \ge 1$ in place of (\ref{eq:budget}),
the same argument gives $P_{\mathrm{rec}}^\star(d) = (B/d)^{n/k}$;
decay in $d$ persists for any $k \ge 1$.

\noindent
\emph{Limitations.} The model assumes per-type independence.
A nation-state adversary that controls infrastructure for
several medium types simultaneously violates this assumption;
in that regime the effective diversity dimension is the
number of \emph{independent control surfaces} rather than the
nominal $d$. Spatial co-location (a Bluetooth tap and a Wi-Fi
tap that share proximity to the client) similarly correlates
capability across types. Strengthening the model to handle
correlated capability is left to future work.

\subsection{Modes Discussed but Not Implemented}
\label{sec:future-modes}

The transparent-proxy architecture admits two further
quantum-resistance modes that we describe here for
completeness but that are out of scope for the current
prototype. Both are well-trodden in the PQC literature; we
sketch how they slot into the proxy and defer empirical
evaluation to future work.

\paragraph*{Mode 3: QKD with SKIP / ETSI~014 key delivery.}
A QKD device co-located with the proxy generates symmetric
keys on a quantum channel and delivers them to the proxy via
a standardized key-delivery interface---either Cisco's Secure
Key Integration Protocol (SKIP)~\cite{cisco-skip}, or
ETSI~GS~QKD~014's REST-based key delivery
API~\cite{etsi-qkd-014}. The proxy then uses
these keys directly to wrap the session key returned to the
client (or as IKEv2 post-quantum preshared keys, in IPsec
deployments). This mode requires dedicated QKD hardware on a
quantum-capable physical link and is therefore most relevant
in metropolitan- or campus-scale deployments where such
infrastructure is already present; the transparent proxy
serves as the integration point that lets non-QKD-capable
clients benefit from a QKD-derived session key. We do not
implement this mode in the present work.

\paragraph*{Mode 4: Classical/PQC hybrid handshake.}
The proxy negotiates a hybrid KEM (e.g., X25519~+~ML-KEM)
with the QKMS, deriving a shared secret that is secure as long
as either the classical or the PQC component holds.
KEMTLS~\cite{kemtls} specifies how a KEM-based handshake (and
its hybrid extensions) is performed on the wire. The hybrid
mode protects against the failure of any single algorithm but,
unlike Mode~2, leaves the surrounding protocol, endpoint set,
and channel structure unchanged. We do not implement this
mode in the present work.
\section{Evaluation}
\label{sec:eval}

We evaluate an implementation of the proposed system on a
prototype deployment, measuring end-to-end key-establishment
latency and decomposing it into its constituent components.

\subsection{Testbed and Methodology}

The testbed consists of five entities deployed as independent
Amazon EC2 instances running Ubuntu 24.04: two clients, a
central QKMS, and two intermediary proxy nodes (one per
client--QKMS pair). Each entity runs the Flask-based
implementation described in Section~\ref{sec:design}. The
implementation provides Wi-Fi and Bluetooth channel backends,
which are used as the heterogeneous transmission media for the
multi-path dispatch in all experiments reported below.

For each configuration we ran 1{,}000 independent trials. We
report the mean of each component on a logarithmic axis to make
the relative magnitudes legible across four orders of magnitude.

For the non-PQ scenario, the proxy/client-side latency is
decomposed into:
\begin{itemize}
\item \textbf{Decryption}: time to decrypt all fragments using
the client's classical private key.
\item \textbf{Reconstruction}: time to validate, re-order, and
concatenate fragments into a session key.
\item \textbf{Network}: aggregate transmission and queuing time
across the channel set.
\end{itemize}
The QKMS-side latency is decomposed into:
\begin{itemize}
\item \textbf{Key Generation}: time to sample the random
session key.
\item \textbf{Key Processing}: time to fragment, optionally
shuffle, and encrypt each fragment under the requester's public
key.
\item \textbf{Network}: time to dispatch fragments across
channels.
\end{itemize}
For the PQ scenario, an additional component---\textbf{PQ KEM},
covering the post-quantum key-encapsulation work used to
establish the wire-line tunnel---is reported on both sides.

\begin{figure}[t]
    \centering
    \includegraphics[width=\linewidth]{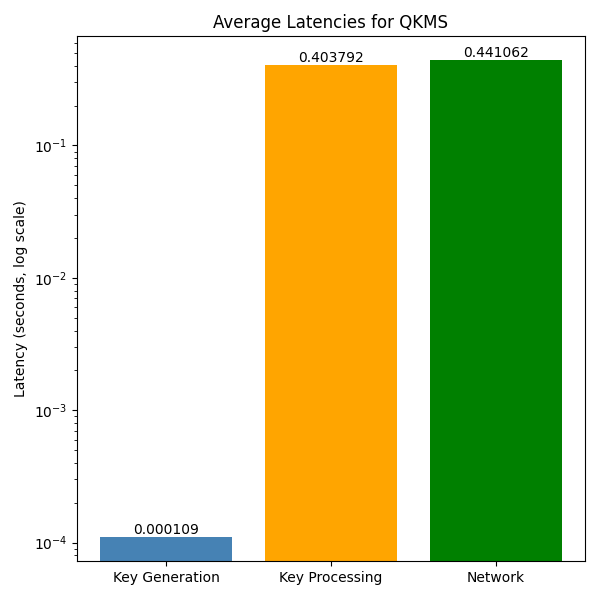}
    \caption{QKMS-side latencies (non-PQ scenario).
    Mean values over 1{,}000 runs, log scale.}
    \label{fig:qkms_mp}
\end{figure}

\begin{figure}[t]
    \centering
    \includegraphics[width=\linewidth]{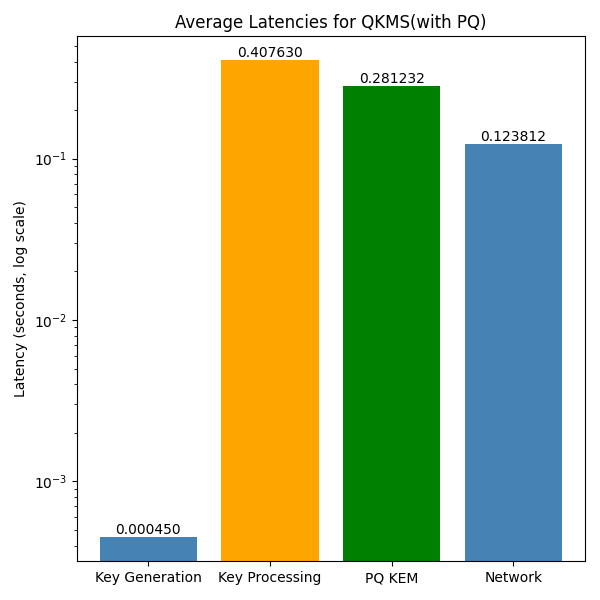}
    \caption{QKMS-side latencies (PQ scenario, multi-path on
    top of a PQ KEM). Mean values over 1{,}000 runs, log scale.}
    \label{fig:qkms_pq}
\end{figure}

\begin{figure}[t]
    \centering
    \includegraphics[width=\linewidth]{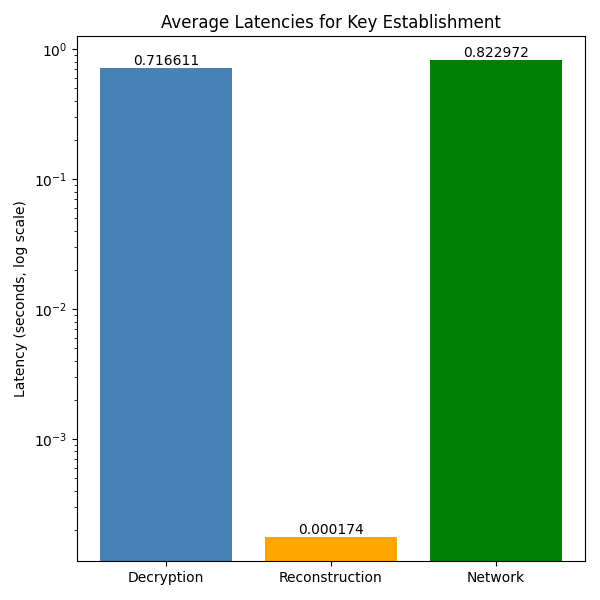}
    \caption{Proxy-side latencies for key establishment over
    multi-path (non-PQ scenario). Mean values over 1{,}000 runs, log scale.}
    \label{fig:proxy_mp}
\end{figure}

\begin{figure}[t]
    \centering
    \includegraphics[width=\linewidth]{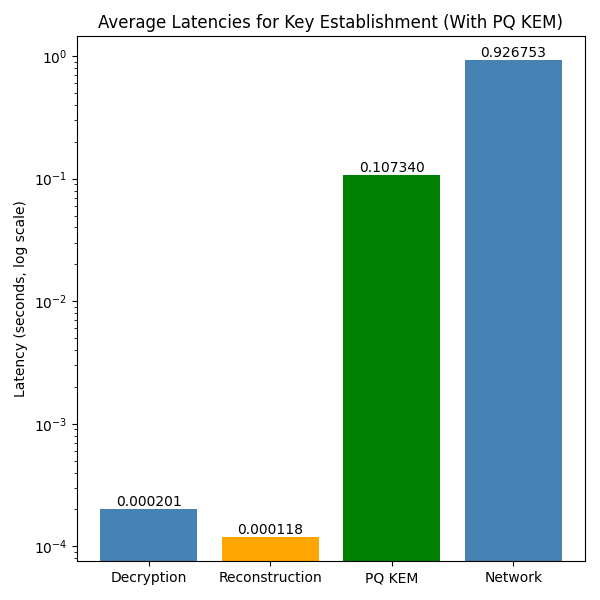}
    \caption{Proxy-side latencies for key establishment over
    multi-path on top of a PQ KEM. Mean values over 1{,}000 runs, log scale.}
    \label{fig:proxy_pq}
\end{figure}

\subsection{Results}

\textbf{QKMS-side, non-PQ (Fig.~\ref{fig:qkms_mp}).}
Key generation is essentially free at $\sim 109\,\mu$s.
Key processing---fragmentation, shuffle, and per-fragment
classical-public-key encryption---averages 404\,ms, and network
dispatch averages 441\,ms. Cryptographic work and network dispatch
are thus comparable on the QKMS side, with both dwarfing key
generation by $\sim$3{,}500$\times$.

\textbf{QKMS-side, PQ (Fig.~\ref{fig:qkms_pq}).}
Adding the PQ KEM contributes 281\,ms of additional latency on
the QKMS side. Key processing is essentially unchanged
(408\,ms). Interestingly, the network component drops to
124\,ms, an artifact of the specific instance pairing used in
this run and within the variance band of the prototype. The
overall QKMS-side cost in the PQ configuration is therefore
roughly comparable to the non-PQ case, with the PQ KEM cost
absorbed largely in parallel with network transit.

\textbf{Proxy-side, non-PQ (Fig.~\ref{fig:proxy_mp}).}
Decryption dominates at 717\,ms (this is the aggregate cost of
classical-public-key decryption for every fragment), while
reconstruction is negligible at $\sim 174\,\mu$s. Network
latency is 823\,ms. The decryption cost reflects the per-fragment
overhead of asymmetric decryption; this can be amortized either
by reducing the number of splits or by replacing per-fragment
asymmetric encryption with symmetric encryption under a one-shot
asymmetric key-wrap.

\textbf{Proxy-side, PQ (Fig.~\ref{fig:proxy_pq}).}
With a PQ KEM in place to establish the wire-line tunnel, the
reported decryption time falls dramatically to
$\sim 201\,\mu$s---this configuration uses the PQ-derived
symmetric key for transport, so per-fragment work is
symmetric rather than asymmetric. Reconstruction remains
negligible at $\sim 118\,\mu$s. PQ KEM contributes 107\,ms,
and network transit dominates at 927\,ms.

\subsection{Observations and Discussion}

Three takeaways emerge.

\emph{First, latency is dominated by network transmission, not
by the multi-path machinery.} In every configuration, the
network component is comparable to or larger than the
cryptographic work. The fragmentation, shuffling, and
per-channel dispatch impose no first-order cost beyond what the
network already imposes; the design is therefore practical to
deploy.

\emph{Second, the PQ overhead is bounded and parallelizable.}
The PQ KEM contributes 100--300\,ms across configurations,
consistent with published ML-KEM benchmarks~\cite{kemtls}. This is
a one-time per-session cost; it does not scale with the number
of fragments and is independent of the multi-path layer.

\emph{Third, per-fragment asymmetric decryption is the principal
optimization target.} The 717\,ms cost in the non-PQ proxy
configuration scales linearly with the split count. A
straightforward optimization is to use a hybrid envelope
construction (asymmetric key-wrap of a per-fragment symmetric
key) so that per-fragment work is symmetric; the PQ-scenario
measurement (201\,$\mu$s decryption when the wire is already
inside a symmetric tunnel) confirms the speedup available.

\subsection{Threats to Validity}

We report mean values across 1{,}000 runs; tail latency on
multi-path schemes can be dominated by the slowest-channel
quantile, and a percentile breakdown (especially the 95th and
99th percentiles of the proxy-side decryption and network
components) would provide a more complete operational picture.
The current evaluation also exercises two channel types
(Wi-Fi and Bluetooth); extending to additional heterogeneous
media---NFC, cellular, Ethernet---would test the diversity
dimension $d$ of Section~\ref{sec:heterogeneity} more
directly and would allow direct empirical estimation of the
adversary capacity vector $\mathbf{c}$ across a larger type
set.
\section{Conclusion}
\label{sec:conclusion}

This paper has presented a session-key establishment scheme that
remains secure against a future quantum adversary without
requiring post-quantum cryptography on the wire. The design
fragments a session key, encrypts each fragment under classical
public-key cryptography, and distributes the fragments across a
heterogeneous ensemble of transmission media. The security
argument does not depend on the quantum-hardness of any single
algorithm; rather, it follows from the practical infeasibility
of an adversary---even a quantum-capable one---simultaneously
compromising every member of a structurally diverse channel
set. We formalized this argument as a per-type capacity model
and proved (Theorem~\ref{thm:decay}) that recovery probability
decays as $(B/d)^n$ in the diversity dimension. The scheme is
deployable in PQ-incapable environments as a stand-alone hedge,
and stacks safely on top of a PQ KEM as a defense-in-depth
layer that survives the failure of the deployed PQC
implementation. Our prototype on AWS EC2 
shows that latency is dominated by ordinary network transit
rather than by multi-path overhead, indicating that the design
is practical at the protocol level. Taken together, the work
suggests that session-key confidentiality against quantum
adversaries is achievable today, on commodity classical media,
in a form that is complementary to the standardization-track
migration to PQC.
\section{Future Work}
\label{sec:future}

The design opens several directions for further investigation.

\noindent\textbf{Extending the channel set.} The current implementation
exercises Wi-Fi and Bluetooth backends; adding NFC, cellular,
and Ethernet backends would let the diversity dimension $d$ of
Theorem~\ref{thm:decay} be varied empirically, and would
support direct measurement of an end-to-end adversary cost---the
capacity-vector entries $c_i$---through red-team exercises on
each medium type.

\noindent\textbf{Hybrid envelope optimization.} Replacing per-fragment
asymmetric decryption with a symmetric envelope under a single
asymmetric (or PQ KEM) key-wrap should reduce proxy-side
decryption from hundreds of milliseconds to microseconds, as
already visible in the PQ-scenario measurement
(Section~\ref{sec:eval}).

\noindent\textbf{Threshold variants.} The current design uses an
all-fragments-required split for tightness; a
$(t,n)$-threshold variant would tolerate partial-channel
unavailability at a small information-theoretic cost.

\noindent\textbf{Strengthening the security model.} The analysis in
Section~\ref{sec:heterogeneity} assumes independent per-type
compromise. Extending the model to correlated capability---a
nation-state adversary that simultaneously controls
infrastructure for several medium types, or attacks that
exploit spatial co-location of media---would yield tighter
bounds for the high-resource adversary regime where the
effective diversity dimension differs from the nominal one.

\noindent\textbf{Implementing Mode 3 and Mode 4.}
Modes~3 (QKD with SKIP / ETSI~014 key delivery) and~4
(classical/PQC hybrid handshakes), described in
Section~\ref{sec:future-modes}, are natural integration points
for the transparent proxy. A full implementation would let the
empirical comparison---latency, throughput, infrastructure
cost---between the four modes be made on the same hardware
substrate, and would surface practical issues (key-buffer
sizing for QKD, hybrid-handshake fragmentation under MTU
limits) that are invisible at the architectural level.

\bibliographystyle{IEEEtran}
\bibliography{bib/main}

\end{document}